\pdfoutput=1
\RequirePackage{ifpdf}
\ifpdf 
\documentclass[pdftex]{sigma}
\else
\documentclass{sigma}
\fi

\numberwithin{equation}{section}

\newtheorem{Theorem}{Theorem}[section]

\newtheorem{Lemma}[Theorem]{Lemma}
\newtheorem{Proposition}[Theorem]{Proposition}
 { \theoremstyle{definition}
\newtheorem{Definition}[Theorem]{Definition}

 }

\def\Z{\mathbb{Z}}
\def\n{\boldsymbol{n}}

\begin{document}

\allowdisplaybreaks

\newcommand{\arXivNumber}{1804.02804}

\renewcommand{\PaperNumber}{065}

\FirstPageHeading

\ShortArticleName{On the Coprimeness Property of Discrete Systems without the Irreducibility Condition}

\ArticleName{On the Coprimeness Property of Discrete Systems\\ without the Irreducibility Condition}

\Author{Masataka KANKI~$^\dag$ and Takafumi MASE~$^\ddag$ and Tetsuji TOKIHIRO~$^\ddag$}

\AuthorNameForHeading{M.~Kanki and T.~Mase and T.~Tokihiro}

\Address{$^\dag$~Department of Mathematics, Kansai University, Japan}
\EmailD{\href{mailto:kanki@kansai-u.ac.jp}{kanki@kansai-u.ac.jp}}

\Address{$^\ddag$~Graduate School of Mathematical Sciences, University of Tokyo, Japan}
\EmailD{\href{mailto:mase@ms.u-tokyo.ac.jp}{mase@ms.u-tokyo.ac.jp}, \href{mailto:toki@ms.u-tokyo.ac.jp}{toki@ms.u-tokyo.ac.jp}}

\ArticleDates{Received April 10, 2018, in final form June 21, 2018; Published online June 27, 2018}

\Abstract{In this article we investigate the coprimeness properties of one and two-dimensional discrete equations, in a situation where the equations are decomposable into several factors of polynomials. After experimenting on a simple equation, we shall focus on some higher power extensions of the Somos-$4$ equation and the (1-dimensional) discrete Toda equation. Our previous results are that all of the equations satisfy the irreducibility and the coprimeness properties if the r.h.s.\ is not factorizable. In this paper we shall prove that the coprimeness property still holds for all of these equations even if the r.h.s.\ is factorizable, although the irreducibility property is no longer satisfied.}

\Keywords{integrability detector; coprimeness; singularity confinement; discrete Toda equation}

\Classification{37K10}

\section{Introduction}
Continuous equations have several established definitions of integrability, e.g., the Frobenius complete integrability for Pfaffian systems, the Liouville--Arnold integrability for Hamiltonian systems. Moreover they possess several useful integrability detectors such as the Painlev\'{e} test, the existence of a Lax representation and the solvability via the inverse scattering method. It is a~natural question to ask whether these schemes apply to discrete equations. Indeed the discrete equations also admit several definitions of integrability for particular types of maps, e.g., an analogue of the Arnold--Liouville integrability for symplectic maps is proposed in~\cite{BRSTZ}. There are also the notions of the multidimensional consistency, the Lax integrability, the Darboux integrability and so on.

Let us review the properties closely related to the integrability of fully-discrete equations. The first discrete integrability detector was the singularity confinement~\cite{SC}, which was proposed as an analogue of the Painlev\'{e} test for ordinary differential equations. A discrete mapping is said to pass the singularity confinement test if an indeterminacy is resolved and the information on the initial values are recovered after a finite number of iterations. The test was successfully applied to construct several discrete Painlev\'{e} equations~\cite{DP} as nonautonomous extensions to integrable discrete mappings. However, it was later discovered that some mappings are not necessarily integrable in the sense of exponential degree growth and the chaotic behaviour of the orbits even if they pass the singularity confinement test~\cite{HV}. It was proposed that the degree growth of the iterates is closely related to the integrability of a discrete equation. The algebraic entropy criterion asserts that, if the growth is exponential (in which case the entropy is positive) then the equation is nonintegrable, while if the growth is of polynomial order (in which case the entropy is equal to zero) then it is integrable~\cite{BV}. Lots of works are done using the algebraic entropy, e.g., a class of two-dimensional lattice models are classified using the entropy in~\cite{HV2} and the growth property is extended to semi-discrete equations~\cite{DV}. Both the singularity confinement and the algebraic entropy have played important roles in studying discrete mappings. It has been a major challenge to overcome several minor but important differences between these two properties.

In this paper we shall employ the zero algebraic entropy criterion as the {\em definition} of the integrability of one-dimensional fully-discrete systems. For equations over a~higher-dimensional lattice, they are defined to be integrable if they have polynomial degree growth of the iterates on the lattice of definition.

Recently a new type of condition related to the discrete integrability has been proposed by the authors to further investigate the singularities of equations in terms of the factorization of each iterate. It is called the coprimeness property and is defined over the field of rational functions of the intial variables~\cite{dKdVSC1}. The coprimeness property is one type of singularity analysis of a discrete equation, which is quite similar to the singularity confinement test and is proved to be satisfied for many of the known discrete integrable systems~\cite{KKMT, dKdVSC2}. The Laurent and the irreducibility properties played important roles in proving the coprimeness property of the given equations. An equation is said to have the Laurent property if every iterate is a Laurent polynomial of the initial variables~\cite{FZ}. Moreover the equation has the irreducibility if every iterate is irreducible as a Laurent polynomial. A Laurent system is defined to have the coprimeness property if every pair of iterates is mutually coprime as Laurent polynomials of the initial variables. The tau-function or its analogue of many discrete integrable systems have these properties~\cite{maserims}. The Laurent property and the degree growth of the iterates are discussed in detail with emphasis on the bilinear forms by one of the authors in~\cite{mase}. However, it is {\em not} our intention to assert that the Laurent property and the coprimeness are integrability criteria, even though these two properties seem to be closely related to the integrability. In fact, it is known that there are many non-integrable Laurent recurrences, one of which we will see later, and moreover, some non-integrable equations have the coprimeness property and can be transformed to Laurent systems~\cite{exhv,2d}. It is worth noting that, since the coprimeness is based on the cancellation of factors, this property can be of help in calculating the algebraic entropy of an equation~\cite{exhv}. We also note that the coprimeness and the irreducibility in themselves are not at all trivial and have drawn an attention of researchers in various areas, e.g., the irreducibility and the coprimeness of the so-called Cauchy--Liouville--Mirimanoff polynomials have a long history~\cite{Beukers, Mirimanoff}.

Let us introduce several approaches to the singularities of the discrete equations related to the coprimeness and the Laurent property. A new property related to the discrete integrability called the Devron property is proposed in \cite{Glick}, whose definition is related to the anti-confined singularities \cite{anticonfinement}. The notion of strong $\tau$-sequence is based on the irreducibility and the coprimeness of Laurent systems~\cite{GP}. An observation on the integrability using the factorization of each iterate is given in~\cite{Viallet}. A~similar approach to the singularities of an equation in terms of the Laurent property using the recursive factorization is found in~\cite{HHKQ, HK}.

At present, one of the difficulties of the coprimeness property is that its proof is too technical in most equations, and we needed to first prove the irreducibility and then attack the coprimeness using the tau-function form and its analogues. It has been a big problem to deal with equations without irreducibility, i.e., when the equation itself decomposes into several factors. As one of the simplest examples we introduce the following recurrence:
\begin{gather}
y_n=\frac{y_{n-1}^r+1}{y_{n-2}}, \label{yrecurrence}
\end{gather}
where $r\ge 2$ is an integer parameter. When $r=2$ it is linearizable and is integrable in the sense of linear degree growth, while, when $r\ge 3$, it is nonintegrable in the sense of exponential degree growth and thus it has positive algebraic entropy~\cite{mase}.

The aim of this paper is to provide new techniques to deeply investigate the coprimeness property, and to provide a proof of the coprimeness that does not depend on the irreducibility property by studying several concrete examples. By following the number of factors in each iterate we shall refine the discussion used to prove the coprimeness for the irreducible equations to the factorizable case. This paper is organized as follows: In Section~\ref{section2}, we explain our new tools, using our example~\eqref{yrecurrence}. In Section~\ref{section3}, equations defined over a~higher-dimensional lattice are studied. In particular we study the coprimeness of the generalized one-dimensional discrete Toda lattice equation, when the equation itself is factorizable. Finally we state without proof the coprimeness of several discrete equations.

\section{Coprimeness-preserving recurrence without the irreducibility}\label{section2}

First, let us study one of the simplest examples of the recurrences that have the coprimeness property but does not satisfy the irreducibility.
The recurrence we study is \eqref{yrecurrence}, where~$y_0$ and~$y_1$ are the initial variables and the parameter $r$ is an integer greater than one. It is known that the equation~\eqref{yrecurrence} has the Laurent property, i.e.,
\begin{gather*}
y_n\in I:=\mathbb{Z}\big[y_0^{\pm}, y_1^{\pm}\big],
\end{gather*}
for all $n\ge 2$ \cite{FZ, mase}. On the other hand, $y_n\in I$ is not irreducible in general: for example~$y_2$ is reducible unless $r=2^m$, $m=0,1,2,\dots$. If we consider this equation on $\mathbb{C}$, $y_2$ always decomposes as
\begin{gather*}
y_2=\frac{1}{y_0}\prod_{j=1}^r \big(y_1-\zeta^{2j-1}\big),
\end{gather*}
where $\zeta=\exp \big(\sqrt{-1}\pi/r\big)$. We shall study the coprimeness property and the factorization of the numerator of~$y_n$ of~\eqref{yrecurrence}. Let us decompose $y_n$ as $y_n=p_n/q_n$ where $p_n$ is a polynomial, $q_n$ is a~monic monomial, and $p_n$ and $q_n$ do not have common factors. First several terms are
\begin{gather*}
p_0=y_0, \qquad q_0=q_1=1,\qquad p_1=y_1,\qquad p_2=y_1^r+1,\qquad q_2=y_0,\\ p_3=\big(y_1^r+1\big)^r+y_0^r,\qquad q_3=y_0^r y_1.
\end{gather*}
It is proved in \cite{mase} that
\begin{gather}
p_n=\frac{p_{n-1}^r+q_{n-1}^r}{p_{n-2}},\qquad q_n=\frac{q_{n-1}^r}{q_{n-2}},\qquad n\ge 4. \label{pnqnform}
\end{gather}
Let us show one lemma:
\begin{Lemma} \label{deglemma1} For $n\ge 2$ we have
\begin{gather*}
\deg p_{n+1}>\deg p_n>\deg q_n>0,
\end{gather*}
and $p_n\big|_{\{y_0=y_1=0\}}=1$. Moreover $q_{n+1}$ is divisible by $q_n$.
\end{Lemma}
Lemma \ref{deglemma1} is readily obtained from equation~\eqref{pnqnform}. The following is the main theorem in this section. Our aim is to introduce new techniques through the proof.
\begin{Theorem} \label{ynproposition}
We have the following properties for the iterate $y_n=p_n/q_n$ of equation~\eqref{yrecurrence}:
\begin{enumerate}\itemsep=0pt
\item[$1.$] For $n\ge 2$, the iterate $p_n$ is factorized into $r$ non-unit irreducible polynomials in $\mathbb{C}[y_0,y_1]$ as
\begin{gather*}
p_n=p_n^{(1)}p_n^{(2)}\cdots p_n^{(r)}.
\end{gather*}
\item[$2.$]
\begin{gather*}
\deg p_n^{(1)}=\deg p_n^{(2)}=\cdots =\deg p_n^{(r)}=\frac{\deg p_n}{r}, \qquad n\ge 2.
\end{gather*}
\item[$3.$] By appropriately rearranging the order of $p_{n}^{(1)},\dots,p_{n}^{(r)}$, we have
\begin{gather*}
p_{n}^{(j)}\, \big|\, \big(y_{n-1} - \zeta^{2j-1}\big), \qquad n\ge 2,
\end{gather*}
in $R:=\mathbb{C}\big[y_0^{\pm},y_1^{\pm}\big]$ for every $j$.
\item[$4.$] $p_n^{(j)}$ and $p_m^{(i)}$ are coprime as polynomials in $\mathbb{C}[y_0,y_1]$ if and only if $(n,j)\neq (m,i)$.
\end{enumerate}
\end{Theorem}
From the fourth property of Theorem~\ref{ynproposition}, we immediately obtain that~$y_n$ and~$y_m$ are coprime if and only if $n\neq m$. Another corollary of Theorem~\ref{ynproposition} is that
\begin{gather*}
p_n^{(j)}=\frac{p_{n-1}-\zeta^{2j-1}q_{n-1}}{p_{n-2}^{(r-j+1)}}, \qquad n\ge 4,
\end{gather*}
which shall be proved in the course of the proof of Theorem~\ref{ynproposition}.

\begin{proof}
Let us prove the four properties by induction on $n$. The fourth property is equivalent to the following two properties: $p_n^{(j)}$ and $p_n^{(i)}$ are coprime if $j\neq i$; $p_m^{(j)}$ and $p_n^{(i)}$ are coprime if $2\le m<n$.

{\bf The case of $\boldsymbol{n=2}$.} It is immediate that $p_2^{(j)}=y_1-\zeta^{2j-1}$. Therefore all the properties are readily obtained.

{\bf The case of $\boldsymbol{n=3}$.} First
\begin{gather}
y_3=\frac{1}{y_1}\prod_{j=1}^r \big(y_2-\zeta^{2j-1}\big). \label{y3factor}
\end{gather}
Since $y_2-\zeta^{2j-1}$ is not a unit in $R$, we have
\begin{gather}
\Omega_R (y_3)\ge r, \label{y3ge}
\end{gather}
where $\Omega_R$ denotes the total number of prime elements (see Appendix~\ref{Append}). The localizations of~$R$ has the following relation
\begin{gather*}
\mathbb{C}\big[y_1^{\pm},y_2^{\pm}\big]\big[y_0^{-1}\big]=R\big[y_2^{-1}\big].
\end{gather*}
Therefore using Lemma \ref{UFDlemma}, we have
\begin{gather*}
r=\Omega_{\mathbb{C} [y_1^{\pm},y_2^{\pm}]}(y_3)\ge \Omega_{\mathbb{C} [y_1^{\pm},y_2^{\pm} ] [y_0^{-1}]}(y_3)=\Omega_{R[y_2^{-1}]}(y_3).
\end{gather*}
Since $y_2$ and $y_3$ are coprime in $R$, by using Lemma~\ref{UFDlemma} again, we have
\begin{gather*}
\Omega_{R[y_2^{-1}]}(y_3)=\Omega_R (y_3).
\end{gather*}
Thus we obtain
\begin{gather}
\Omega_R(y_3)\le r. \label{y3le}
\end{gather}
From equations \eqref{y3ge} and \eqref{y3le}, we have
\begin{gather*}
\Omega_R(y_3)=r.
\end{gather*}
Therefore the decomposition of $y_3$ into irreducible elements in $R$ is equation \eqref{y3factor} itself, since each $y_2-\zeta^{2j-1}$ is a non-unit irreducible element in~$R$. Thus the first property of Theorem~\ref{ynproposition} is proved for $n=3$. From equation \eqref{pnqnform}, we have
\begin{gather*}
p_3=y_0^r\prod_{j=1}^r \big(y_2-\zeta^{2j-1}\big),\qquad p_3^{(j)}=y_0\big(y_2-\zeta^{2j-1}\big).
\end{gather*}
The degree of $p_3^{(j)}$ is clearly independent of $j$, thus the second property is proved. The third property is trivial from the explicit form of~$p_3^{(j)}$. Finally let us prove the fourth property. Since~$y_3$ and~$y_2$ are mutually coprime in $R$, $p_3^{(j)}$ and $p_2^{(i)}$ are mutually coprime. We have
\begin{gather*}
p_3^{(j)}-p_3^{(i)}=y_0\big(\zeta^{2i-1}-\zeta^{2j-1}\big),
\end{gather*}
which is a unit in $R$. Therefore $p_3^{(i)}$ and $p_3^{(j)}$ are mutually coprime in~$R$ if $j\neq i$. Moreover, since~$p_3^{(i)}$ and~$p_3^{(j)}$ do not have monomial factors, they are mutually coprime in~$\mathbb{C}[y_0,y_1]$.

{\bf The case of $\boldsymbol{n\ge 4}$.}
 First let us prove that $y_n$ is coprime with $y_{n-2}$ in $R$. Note that $y_n$ is coprime with $y_{n-1}$ from the form of~\eqref{yrecurrence}. From the induction hypothesis, the decomposition of~$y_{n-2}$ into irreducible elements is
\begin{gather*}
y_{n-2}=u p_{n-2}^{(1)} p_{n-2}^{(2)} \cdots p_{n-2}^{(r)},
\end{gather*}
where $u$ is a unit in $R$. We need to prove that $p_{n-2}^{(j)}$ does not divide $y_n$ for any $j$. A direct calculation shows
\begin{gather*}
y_n=\frac{y_{n-1}^r+1}{y_{n-2}}=\frac{1+y_{n-3}^r}{y_{n-3}^r y_{n-2}}+y_{n-2}f
=\frac{y_{n-4}}{y_{n-3}^r}+y_{n-2} f
\end{gather*}
for some $f\in R\big[y_{n-3}^{-1}\big]$. Since $p_{n-2}^{(j)}$ is coprime with $y_{n-4}$ and $y_{n-3}$, there exist constants $\alpha, \beta\in\mathbb{C}^{\times}$ such that
\begin{gather*}
p_{n-2}^{(j)}\big|_{y_0=\alpha, \, y_1=\beta}=0,\qquad y_{n-4}\big|_{y_0=\alpha,\, y_1=\beta}\neq 0,\qquad y_{n-3}\big|_{y_0=\alpha, \, y_1=\beta}\neq 0.
\end{gather*}
In this setting, we have $y_n\big|_{y_0=\alpha,\, y_1=\beta}\neq 0$, which indicates that $y_n$ is not divisible by $p_{n-2}^{(j)}$.

We prove that $y_{n-1}-\zeta^{2j-1}$ is divisible by $p_{n-2}^{(r-j+1)}$ in $R$. Since $p_{n-2}^{(r-j+1)}$ is coprime with $y_{n-3}$, it is sufficient to prove that $y_{n-1}-\zeta^{2j-1}$ is divisible by $p_{n-2}^{(r-j+1)}$ in $R\big[y_{n-3}^{-1}\big]$. We have
\begin{gather*}
y_{n-1}-\zeta^{2j-1}=\frac{y_{n-2}^r+1}{y_{n-3}}-\zeta^{2j-1}\equiv \frac{1}{y_{n-3}}-\zeta^{2j-1}=-\frac{\zeta^{2j-1}}{y_{n-3}}\big(y_{n-3}-\zeta^{1-2j}\big)\equiv 0,
\end{gather*}
where $\equiv$ indicates a equivalence modulo $p_{n-2}^{(r-j+1)}$.

We shall prove that
\begin{gather}
\Omega_R (y_n)=r. \label{omegaofyn}
\end{gather}
First we prove that $\Omega_R (y_n)\ge r$. We have
\begin{gather*}
y_n=q_{n-2}\prod_{j=1}^r \frac{y_{n-1}-\zeta^{2j-1}}{p_{n-2}^{(r-j+1)}},
\end{gather*}
where
\begin{gather*}
\frac{y_{n-1}-\zeta^{2j-1}}{p_{n-2}^{(r-j+1)}}\in R,
\end{gather*}
from the previous step. Moreover, each factor
\begin{gather*}
\frac{y_{n-1}-\zeta^{2j-1}}{p_{n-2}^{(r-j+1)}}
\end{gather*}
is not a unit in $R$, since
\begin{gather*}
\deg_L \left( \frac{y_{n-1}-\zeta^{2j-1}}{p_{n-2}^{(r-j+1)}} \right)=\deg_L \big(y_{n-1}-\zeta^{2j-1}\big)-\deg p_{n-2}^{(r-j+1)}=\deg p_{n-1}-\frac{\deg p_{n-2}}{r}>0,
\end{gather*}
where $\deg_L$ denotes the degree as a Laurent polynomial (see Appendix~\ref{Append}). Here we have used Lemma~\ref{deglemma1} to prove the last inequality. Therefore $\Omega_R(y_n)\ge r$.

Next we prove that $\Omega_R (y_n)\le r$ by using a relation on the localizations of~$R$. First, from
\begin{gather*}
y_n=\frac{1}{y_{n-2}}\prod_{j=1}^r \big(y_{n-1}-\zeta^{2j-1}\big),
\end{gather*}
we have
\begin{gather*}
\Omega_{\mathbb{C}[y_{n-2}^{\pm}, y_{n-1}^{\pm}]}(y_n)=r.
\end{gather*}
We have following relations on the localizations of $R$:
\begin{gather*}
R\big[y_{n-2}^{-1},y_{n-1}^{-1}\big]=\mathbb{C}\big[y_{n-2}^{\pm},y_{n-1}^{\pm}\big]\big[y_0^{-1},y_1^{-1}\big].
\end{gather*}
Therefore
\begin{gather*}
\Omega_{R[y_{n-2}^{-1},y_{n-1}^{-1}] }\le r.
\end{gather*}
Since we have already proved that $y_n$ is coprime with both $y_{n-1}$ and $y_{n-2}$, it follows from Lemma~\ref{UFDlemma} that $\Omega_R(y_n)=r$. From~\eqref{omegaofyn} we have
\begin{gather*}
p_n^{(j)}=\frac{q_{n-1}\big(y_{n-1}-\zeta^{2j-1}\big)}{p_{n-2}^{(r-j+1)}}=\frac{p_{n-1}-\zeta^{2j-1}q_{n-1}}{p_{n-2}^{(r-j+1)}},
\end{gather*}
from which we can conclude that $\deg p_n^{(j)}$ does not depend on $j$. Now the second and the third properties are proved.

Finally we prove the fourth property of Theorem~\ref{ynproposition}. For $2\le m<n$, the two iterates~$p_m^{(i)}$ and~$p_n^{(j)}$ are both non-unit irreducible polynomials with distinct degrees (note that $\deg p_m^{(i)}\neq \deg p_n^{(j)}$ from Lemma~\ref{deglemma1}). Therefore they are mutually coprime. Next we prove the mutual coprimeness of~$p_n^{(i)}$ and~$p_n^{(j)}$ for $i\neq j$ as polynomials. Since
\begin{gather*}
p_n^{(j)}p_{n-2}^{(r-j+1)}-p_n^{(i)}p_{n-2}^{(r-i+1)}=q_{n-1}\big(\zeta^{2i-1}-\zeta^{2j-1}\big),
\end{gather*}
$p_n^{(j)}$ and $p_n^{(i)}$ are mutually coprime in $R$. Moreover, both of them do not have monomial factors, and thus they are coprime as polynomials.
\end{proof}

\section[One-dimensional discrete Toda type equations without the irreducibility]{One-dimensional discrete Toda type equations\\ without the irreducibility}\label{section3}

In our previous work, we have introduced the following `pseudo-integrable' extension to the one-dimensional discrete Toda equation
\begin{gather}
\tau_{t,n}=\frac{1}{\tau_{t-2,n}}\big( \tau_{t-1,n+1}^{M}\tau_{t-1,n-1}^{L}+\tau_{t-1,n}^{K} \big), \label{cp1dtoda-prev}
\end{gather}
where $M$, $L$, $K$ are some positive integers \cite{KKMT}. When $(M,L,K)=(1,1,2)$ the equation \eqref{cp1dtoda-prev} is the one-dimensional discrete Toda equation. It should be noted that the equation \eqref{cp1dtoda-prev} was first introduced as the number wall and its Laurent property has been proved in~\cite{FZ}. One of our results is as follows
\begin{Proposition}[{\cite[Proposition 5.1]{KKMT}}] \label{previousprop} Every iterate $\tau_{t,n}$ is irreducible and mutually coprime in
\begin{gather*}
\mathbb{Z}\big[\tau_{0,n}^{\pm}, \tau_{1,n}^{\pm}\, |\, n\in\mathbb{Z}\big]
\end{gather*}
on condition that $\mathrm{GCD}(K,L,M)$ is a non-negative power of~$2$.
\end{Proposition}
However, we did not present its proof there. Here for the reader of~\cite{KKMT} we need to remark that in the statement of Proposition~5.1 in~\cite{KKMT}, we have mistakenly omitted the condition that ``$\mathrm{GCD}(K,L,M)$ is a non-negative power of~$2$''.
The proof of proposition \ref{previousprop} depends on the fact that the r.h.s.~$\big(X^M Y^L+Z^K\big)$ is irreducible as a polynomial in $\mathbb{Z}[X,Y,Z]$ if and only if $\mathrm{GCD}(K,L,M)=2^q$ for some $q\ge 0$.
In this article we focus on the equation \eqref{cp1dtoda-prev} with a factorizable r.h.s. We shall prove that, even if the r.h.s.\ is factorizable in $\mathbb{Z}[X,Y,Z]$ (e.g., $X^3Y^3+Z^3$), the coprimeness property still holds. Let us investigate
\begin{gather}
\tau_{t,n}=\frac{1}{\tau_{t-2,n}}\big( \tau_{t-1,n+1}^{rm}\tau_{t-1,n-1}^{rl}+\tau_{t-1,n}^{rk} \big), \label{cp1dtoda}
\end{gather}
where $r\ge 2$ and $\mathrm{GCD}(m,l,k)=1$. The irreducibility of its iterates is not satisfied, since the r.h.s.\ factorizes as
\begin{gather*}
\tau_{t,n}=\frac{1}{\tau_{t-2,n}} \prod_{j=1}^r \big( \tau_{t-1,n+1}^{m}\tau_{t-1,n-1}^{l} - \zeta^{2j-1}\tau_{t-1,n}^{k} \big),
\end{gather*}
where $\zeta=\exp(\sqrt{-1}\pi/r)$. However, the coprimeness property is satisfied as stated below in Theorem~\ref{mainthm}. The set of initial variables of the equation is
\begin{gather*}
\left\{ \tau_{0,n},\tau_{1,n} \, |\, n\in \mathbb{Z} \right\},
\end{gather*}
and we consider the evolution of equation \eqref{cp1dtoda} towards $t\ge 2$.
Let us define
\begin{gather*}
R_0=\mathbb{C}\big[ \big\{\tau_{0,n}^{\pm},\tau_{1,n}^{\pm}\big\}_{n\in\mathbb{Z}} \big]
\end{gather*}
and discuss the coprimeness of distinct two iterates over $R_0$. We follow the procedure used in the previous section. Let us note that $X^{m}Y^{l}-\zeta^{2j-1}Z^{k}$ is irreducible as a polynomial in $\mathbb{C}[X,Y,Z]$ for every integer $j$.

\begin{Theorem} \label{mainthm}
Let us write the iterates of equation \eqref{cp1dtoda} as
\begin{gather*}
\tau_{t,n}=\frac{p_{t,n}}{q_{t,n}},
\end{gather*}
where $p_{t,n}$, $q_{t,n}$ are polynomials, $q_{t,n}$ is a monic monomial, and $p_{t,n}$ and $q_{t,n}$ do not have common factors. Then we have the following four properties:
\begin{enumerate}\itemsep=0pt
\item[$1.$] For $t\ge 2$, $p_{t,n}$ is factorized into the product of $r$ non-unit irreducible polynomials in $R_0$ as
\begin{gather*}
p_{t,n}=p_{t,n}^{(1)}p_{t,n}^{(2)}\cdots p_{t,n}^{(r)}.
\end{gather*}
\item[$2.$] For $i=1,2,\dots, r$, we have
\begin{gather*}
\deg p_{t,n}^{(i)}=\frac{\deg p_{t,n}}{r}, \qquad t\ge 2.
\end{gather*}
\item[$3.$] By appropriately rearranging the order of $p_{t,n}^{(1)},\dots,p_{t,n}^{(r)}$, we have
\begin{gather*}
p_{t,n}^{(j)}\, \big|\, \big(\tau_{t-1,n+1}^{m}\tau_{t-1,n-1}^{l} - \zeta^{2j-1}\tau_{t-1,n}^{k}\big),\qquad t\ge 2,
\end{gather*}
for every $j$.
\item[$4.$] Two iterates $p_{t,n}^{(i)}$ and $p_{s,n'}^{(j)}$ are coprime unless $(i,t,n)=(j,s,n')$.

In particular, if $(t,n)\neq (s,n')$, two iterates $\tau_{t,n}$ and $\tau_{s,n'}$ are coprime as Laurent polynomials.
\end{enumerate}
\end{Theorem}

Let us remark that, from Theorem \ref{mainthm}, $p_{t,n}^{(j)}$ is recursively defined as
\begin{gather}
p_{t,n}^{(j)}=\frac{p_{t-1,n}-\zeta^{2j-1} q_{t-1,n}}{p_{t-2,n}^{(r-j+1)}}, \label{ptnj}
\end{gather}
which is derived in the course of proving Theorem~\ref{mainthm}.

The proof is done by induction. Let us prepare several lemmas to prove Theorem~\ref{mainthm}. Let us assume in the following lemmas that properties $1$ through $4$ are satisfied for $p_{s,n}$ with $s\le t-1$.

\begin{Lemma} \label{dividelemma}
The term $p_{t-2,n}^{(r-j+1)}$ divides $\big(\tau_{t-1,n+1}^{m}\tau_{t-1,n-1}^{l} - \zeta^{2j-1}\tau_{t-1,n}^{k}\big)$ in the ring of Laurent polynomials $\mathbb{C}\big[\big\{\tau_{0,n}^{\pm},\tau_{1,n}^{\pm}\big\}_{n\in\mathbb{Z}}\big]$.
\end{Lemma}

\begin{proof} Every calculation shall be done modulo $p_{t-2,n}^{(r-j+1)}$. We have
\begin{gather}
 \big(\tau_{t-1,n+1}^{m}\tau_{t-1,n-1}^{l} - \zeta^{2j-1}\tau_{t-1,n}^{k}\big) \nonumber\\
 \qquad{} \equiv \left( \frac{\tau_{t-2,n+1}^{rk}}{\tau_{t-3,n+1}} \right)^{m}\left( \frac{\tau_{t-2,n-1}^{rk}}{\tau_{t-3,n-1}} \right)^{l}-\zeta^{2j-1} \left( \frac{\tau_{t-2,n+1}^{rm}\tau_{t-2,n-1}^{rl}}{\tau_{t-3,n}} \right)^{k} \nonumber \\
\qquad{} \equiv \frac{\tau_{t-2,n+1}^{rkm} \tau_{t-2,n-1}^{rlk} }{ \tau_{t-3,n+1}^m \tau_{t-3,n-1}^l \tau_{t-3,n}^k }\big({-}\zeta^{2j-1}\big)\big(\tau_{t-3,n+1}^m \tau_{t-3,n-1}^l -\zeta^{2r-2j+1}\tau_{t-3,n}^k\big). \label{lemma1eq}
\end{gather}
By the induction hypothesis (the third property in Theorem~\ref{mainthm}) that $p_{t-2,n}^{(r-j+1)}$ divides
\begin{gather*} \tau_{t-3,n+1}^m \tau_{t-3,n-1}^l -\zeta^{2(r-j+1)-1}\tau_{t-3,n}^k,\end{gather*} we conclude that equation~\eqref{lemma1eq} is equal to~$0$ modulo $p_{t-2,n}^{(r-j+1)}$.
\end{proof}

\begin{Lemma} \label{lemmataut-4}
\begin{gather*}
\tau_{t-1,n+1}^{rm}\tau_{t-1,n-1}^{rl}+\tau_{t-1,n}^{rk}\equiv \frac{\tau_{t-2,n+1}^{r^2 km}\tau_{t-2,n-1}^{r^2 k l}}{\tau_{t-3,n+1}^{rm} \tau_{t-3,n-1}^{rl} \tau_{t-3,n}^{rk}}\tau_{t-4,n}\tau_{t-2,n},
\end{gather*}
where $\equiv$ is taken as modulo $\tau_{t-2,n}^2$.
\end{Lemma}
It is immediately obtained by a direct computation using
\begin{gather*} \tau_{t-2,n} \tau_{t-4,n}=\tau_{t-3,n+1}^{rm}\tau_{t-3,n-1}^{rl}+\tau_{t-3,n}^{rk}.\end{gather*}

\begin{Lemma}\label{lemma53}
For $t\ge 4$, we have the following properties on $p_{t,n}$ and $q_{t,n}$:
\begin{enumerate}\itemsep=0pt
\item[$(a)$] \begin{gather*}
q_{t,n}=\frac{{\rm LCM} \big(q_{t-1,n+1}^{rm} q_{t-1,n-1}^{rl}, q_{t-1,n}^{rk}\big)}{q_{t-2,n}},\qquad p_{t,n}=\frac{h_n p_{t-1,n+1}^{rm} p_{t-1,n-1}^{rl} +h_n' p_{t-1,n}^{rk}}{p_{t-2,n}},
\end{gather*}
where
\begin{gather*}
h_n=\frac{{\rm LCM}\big(q_{t-1,n+1}^{rm}q_{t-1,n-1}^{rl}, q_{t-1,n}^{rk}\big)}{q_{t-1,n+1}^{rm} q_{t-1,n-1}^{rl}},\qquad h_n'=\frac{{\rm LCM}\big(q_{t-1,n+1}^{rm}q_{t-1,n-1}^{rl}, q_{t-1,n}^{rk}\big)}{q_{t-1,n}^{rk}}.
\end{gather*}

\item[$(b)$] The iterate $q_{t,n}$ is divisible by all the three iterates $q_{t-1,n}$, $q_{t-1,n-1}$ and $q_{t-1,n+1}$.

\item[$(c)$] $p_{t,n}$, $t\ge 2$, is not divisible by any of the initial variables $\tau_{0,m}$, $\tau_{1,m}$, $m\in\mathbb{Z}$.\end{enumerate}
\end{Lemma}

\begin{proof} The discussion is similar to the one used to prove~\eqref{pnqnform}. Details are omitted in the paper.
\end{proof}

Let us remark that the property (c) in Lemma \ref{lemma53} is needed to assure that $p_{t-2,n}$ does not have a monomial factor when $p_{t,n}=\prod\limits_{j=1}^r p_{t,n}^{(j)}$ is factorized as in equation~\eqref{ptnj}. From this observation we conclude that the r.h.s.\ of the second equation in (a) (the recurrence relation of~$p_{t,n}$) is not only a Laurent polynomial in~$R_0$, but also a polynomial.

\begin{Lemma} \label{lemmapqdeg}
We have the following two inequalities for the degrees of the iterates:
\begin{gather}
\deg p_{t,n} > \deg q_{t,n}>0, \label{degpq}\\
\deg p_{t+1,n} > \deg p_{t,n}. \label{degpp}
\end{gather}
\end{Lemma}

\begin{proof} Let us prove \eqref{degpq} by induction. We have
\begin{gather}
\frac{p_{t,n}}{q_{t,n}}=\frac{q_{t-2,n}}{p_{t-2,n}}\left( \frac{p_{t-1,n+1}^{rm}}{q_{t-1,n+1}^{rm}}\frac{p_{t-1,n-1}^{rl}}{q_{t-1,n-1}^{rl}}+\frac{p_{t-1,n}^{rk}}{q_{t-1,n}^{rk}} \right). \label{lemma510eq}
\end{gather}
Since the degrees of $p_{t,n}$ and $q_{t,n}$ are independent of $n$, we can well-define
\begin{gather*}
d_t:=\deg p_{t,n}-\deg q_{t,n}.
\end{gather*}
From~\eqref{lemma510eq}, we have
\begin{gather*}
d_t\ge \max\left( r(m+l)d_{t-1}, rk d_{t-1}\right)-d_{t-2}=r \max(m+l,k)d_{t-1}-d_{t-2}.
\end{gather*}
Since $r\ge 2$, under the induction hypothesis $d_{t-1}>d_{t-2}$, we have
\begin{gather*}
d_t\ge 4d_{t-1}-d_{t-2}>3d_{t-1}>d_{t-1}.
\end{gather*}
The equation \eqref{degpp} is readily obtained using Lemma~\ref{lemma53}.\end{proof}

\begin{Lemma} \label{lemmadegtau} Let us define the degree of $\tau_{t,n}$ as a rational function of the initial variab\-les~$\tau_{0,n}$, $\tau_{1,n}$, $n\in\mathbb{Z}$, as $\deg \tau_{t,n}$. We have that $\deg \tau_{t,n}$ does not depend on~$n$, and is monotonously increasing with respect to~$t$.
\end{Lemma}
Lemma~\ref{lemmadegtau} is immediately obtained from Lemma~\ref{lemmapqdeg}.

\begin{Lemma} \label{degmlk} The degree
\begin{gather*}
\deg_L \big(\tau_{t,n}^m \tau_{t,n-2}^l -\zeta^{2j-1}\tau_{t,n-1}^k\big)
\end{gather*}
is independent of $j\in\mathbb{Z}$.
\end{Lemma}
It is sufficient to show that there is no cancellation of the highest (and the lowest) terms in~$\tau_{t,n}^m \tau_{t,n-2}^l$ and~$\zeta^{2j-1}\tau_{t,n-1}^k$, which shall be proved for $\zeta^{2j-1}\not\in \mathbb{Q}$ and $\zeta^{2j-1}=-1$ respectively. Details are omitted here. The following Lemma~\ref{coprimethree} is the key to the proof of Theorem~\ref{mainthm}.
\begin{Lemma} \label{coprimethree}
The iterate $\tau_{t-2,n}$ is coprime with every iterate in $\{ \tau_{t-3,n}, \tau_{t-4,n} \}_{n\in\mathbb{Z}}$.
\end{Lemma}

\begin{proof}
Let us split the proof in the two cases: $t=3$ and $t\ge 4$.

{\bf The case of $\boldsymbol{t=3}$.} It is sufficient to prove that $\tau_{3,n}$ and $\tau_{2,m}$ are coprime with each other for arbitrary $n,m\in\mathbb{Z}$. For this purpose it is sufficient to show that $\tau_{3,n}$ is not divisible by~$p_{2,m}^{(j)}$ for all $m\in\mathbb{Z}$ and $j\in\{1,2,\dots,r\}$. Since $\tau_{2,m}$ with $m\neq n\pm 1,n$ does not share a factor with~$\tau_{3,n}$, $p_{2,m}^{(j)}$ does not divide $\tau_{3,n}$ for $m\neq n\pm 1,n$. We shall prove the cases of $m=n\pm 1, n$. If $p_{2,n}^{(j)}=0$, then we have $\tau_{3,n}=\frac{\tau_{2,n-1}^{rl} \tau_{2,n+1}^{rm}}{\tau_{1,n}}$. Thus it is possible for us to assign suitable values to the initial data so that we have $p_{2,n}^{(j)}=0$ and at the same time $\tau_{3,n}\neq 0$.\footnote{Since $\tau_{2,m}$'s are coprime with each other, it is possible to achieve that $p_{2,m}^{(j)}=0$ and at the same time $\tau_{2,n\pm 1}\neq 0$.} Therefore $p_{2,n}^{(j)}$ does not divide~$\tau_{3,n}$. The case of $m=n+1$ is proved in the same manner since $\tau_{3,n}=\frac{\tau_{2,n}^{rk}}{\tau_{1,n}}$ if~$p_{2,n+1}^{(j)}=0$. The case of $m=n-1$ is readily obtained from the symmetries of the equation.

{\bf The case of $\boldsymbol{t\ge 4}$.} When we calculate the iterate $\tau_{t-2,n}$, we use the following eight iterates in $\{\tau_{t-3,m}, \tau_{t-4,m}\, |\, m\in\mathbb{Z}\}$:
\begin{gather}
\tau_{t-3,n}, \ \tau_{t-3,n\pm1}, \ \tau_{t-4,n}, \ \tau_{t-4,n\pm 1}, \ \tau_{t-4,n\pm 2}. \label{eightiterate}
\end{gather}
It is sufficient to prove the coprimeness of these eight iterates with~$\tau_{t-2,n}$. We shall prove that~$\tau_{t-2,n}$ is not divisible by any~$p_{s,m}^{(j)}$ that is a factor of~\eqref{eightiterate}.

{\bf $\boldsymbol{\tau_{t-2,n}}$ is not divisible by $\boldsymbol{p_{t-4,n}^{(j)}}$.} From Lemma~\ref{lemmataut-4}, we obtain
\begin{gather*}
\tau_{t-2,n}\tau_{t-4,n}\equiv \frac{\tau_{t-4,n+1}^{r^2 km}\tau_{t-4,n-1}^{r^2 k l}}{\tau_{t-5,n+1}^{rm} \tau_{t-5,n-1}^{rl} \tau_{t-5,n}^{rk}}\tau_{t-6,n}\tau_{t-4,n},
\end{gather*}
where $\equiv$ is taken as modulo $\tau_{t-4,n}^2$. Dividing the both sides by $\tau_{t-4,n}$, and taking modulo $p_{t-4,n}^{(j)}$ we have that $p_{t-4,n}^{(j)}$ does not divide $\tau_{t-2,n}$. On the other hand $p_{t-4,n}^{(j)}$ is irreducible from the induction hypothesis. Thus~$\tau_{t-2,n}$ is coprime with~$\tau_{t-4,n}$.

{\bf $\boldsymbol{\tau_{t-2,n}}$ is not divisible by $\boldsymbol{p_{t-3, n\pm 1}^{(j)}=0}$.} It is sufficient to show that for a fixed integer $j\in\{1,2,\dots,r\}$, there exists a set of non-zero initial values such that $\tau_{t-2,n}\neq 0$ and at the same time $p_{t-3,n-1}^{(j)}=0$.

The initial variable $X:=\tau_{1,n-t+3}$ is included in the expansion of $\tau_{t-3,n-1}$, but is not included in either of the expansion of $\tau_{t-3,n}$ or $\tau_{t-4,n}$. The factor $p_{t-3,n-1}^{(j)}$ is a polynomial of $X$ but is not a monomial with respect to $X$. Let us express $p_{t-3,n-1}^{(j)}$ as
\begin{gather*}
p_{n-3,t-1}^{(j)}=\sum_k a_k(\boldsymbol{\tau}) X^k,
\end{gather*}
where $a_k(\boldsymbol{\tau})$ is a Laurent polynomial of the initial variables other than $X$. There exist non-zero initial values $\widetilde{\boldsymbol{\tau}}$ such that $a_k(\widetilde{\boldsymbol{\tau}}) \neq 0$ ($\forall\, k$), $\tau_{t-3,n}\neq 0$ and $\tau_{t-4,n}\neq 0$. Since $p_{n-3,t-1}^{(j)}$ is not a~monomial of $X$, the algebraic equation with respect to $X$
\begin{gather*}
\sum_k a_k(\widetilde{\boldsymbol{\tau}}) X^k=0
\end{gather*}
has a non-zero root $\widetilde{X}$. By taking $\widetilde{\boldsymbol{\tau}}$, $\widetilde{X}$ as the initial values we have $\tau_{t-2,n}=\tau_{t-3,n}^{rk}/\tau_{t-4,n}$, which is non-zero by construction.

Proof of the case $p_{t-3, n+1}^{(j)}=0$ is done in the same manner.

{\bf $\boldsymbol{\tau_{t-2,n}}$ is not divisible by $\boldsymbol{p_{t-3,n}^{(j)}=0}$.} Let us fix an integer $j\in\{1,2,\dots,r\}$ and show that we attain $\tau_{t-2,n}\neq 0$ and $p_{t-3,n}^{(j)}=0$ simultaneously for a~set of initial values. When $p_{t-3,n}^{(j)}=0$ we have $\tau_{t-3,n}=0$ and
\begin{gather*}
\tau_{t-2,n}=\frac{\tau_{t-3,n-1}^{rl}\tau_{t-3,n+1}^{rm}}{\tau_{t-4,n}}.
\end{gather*}
The iterate $\tau_{t-3,n}$ depends only on
\begin{gather*}
\tau_{0,m},\quad m\ge n-t+5,\qquad \tau_{1,m'}, \quad m'\ge n-t+4,
\end{gather*}
among the initial variables. Since the iterate $\tau_{t-2,n}$ is a polynomial of $X=\tau_{1,n-t+3}$, by choosing a value of~$X$ avoiding the zeros of $\tau_{t-2,n}$ as a polynomial of~$X$, $\tau_{t-2,n}$ becomes non-zero, under the condition that $p_{t-3,n}^{(j)}=0$.

{\bf $\boldsymbol{\tau_{t-2,n}}$ is not divisible by $\boldsymbol{p_{t-4,n-2}^{(j)}=0}$.} We shall prove that for a fixed integer $j\in\{1,2,\dots, r\}$ there exists at least one set of initial values such that $\tau_{t-2,n}\neq 0$ and $p_{t-4,n-2}^{(j)}=0$. When $p_{t-4,n-2}^{(j)}=0$, naturally $\tau_{t-4,n-2}=0$ and we have
\begin{gather*}
\tau_{t-3,n-1}=\frac{\tau_{t-4,n-1}^{rk}}{\tau_{t-5,n-1}},\qquad \tau_{t-2,n}=\frac{ \tau_{t-3,n}^{rk}\tau_{t-5,n-1}^{rl} + \tau_{t-3.n+1}^{rm}\tau_{t-4,n-1}^{r^2 kl } }{\tau_{t-4,n}\tau_{t-5,n-1}^{rl}}.
\end{gather*}
The eight iterates $\{\tau_{0,t+n-\alpha-1}, \tau_{1,t+n-\alpha}\}$, $\alpha=3,4,5,6$, are used to define $\tau_{t-2,n}$, but not used to define~$\tau_{t-4,n-2}$ (or~$p_{t-4,n-2}^{(j)}$). The iterate $\tau_{t-3,n+1}$ becomes~$0$ by choosing the topmost (in the~$t$-$n$ plane) two iterates $\tau_{0,t+n-4}$ and $Y=\tau_{1,t+n-3}$ among the above eight iterates appropriately, since $\tau_{t-3,n+1}=0$ is written down as a polynomial of~$Y$. In this case we have $\tau_{t-3,n+1}=\tau_{t-3,n}^{rk}/\tau_{t-4,n}$, which becomes non-zero by assigning appropriate values to the remaining six iterates $\tau_{0,t+n-\alpha-1}$ and $\tau_{1,t+n-\alpha}$, $\alpha=4,5,6$. Therefore we can achieve $\tau_{t-2,n}\neq 0$ preserving the condition that~$p_{t-4,n-2}^{(j)}=0$.

The proof of the cases $p_{t-4,n-1}^{(j)}$, $p_{t-4,n+1}^{(j)}$ and $p_{t-4,n+2}^{(j)}$ can be done in the same manner since the three iterates $\tau_{t-4,n-1}$, $\tau_{t-4,n+1}$, $\tau_{t-4,n+2}$ are defined {\em without} using the iterate $X=\tau_{1,n-t+3}$.
\end{proof}

\begin{proof}[Proof of Theorem \ref{mainthm}]

{\bf The case of $\boldsymbol{t=2}$.}
\begin{gather*}
\tau_{2,n}=\frac{1}{\tau_{0,n}}\prod_{j=1}^r \big( \tau_{1,n+1}^m \tau_{1,n-1}^l -\zeta^{2j-1} \tau_{1,n}^k \big)
\end{gather*}
trivially satisfies all the four conditions in Theorem~\ref{mainthm} if we take
\begin{gather*}
q_{2,n}=\tau_{0,n},\qquad p_{2,n}^{(j)}=\tau_{1,n+1}^m \tau_{1,n-1}^l -\zeta^{2j-1} \tau_{1,n}^k.
\end{gather*}

{\bf The case of $\boldsymbol{t=3}$.} Let us prove that
\begin{gather}
\Omega_{R_0}(\tau_{3,n})=r, \label{tau3n}
\end{gather}
where $\Omega_{*}$ specifies the number of prime elements in a unique factorization domain ``$*$'' as explained in the appendix. From the induction hypothesis, we have
\begin{gather*}
\Omega_{R_1}(\tau_{3,n})=r.
\end{gather*}
Since we have the following equality between two localized rings:
\begin{gather*}
R_0 \big[\big\{\tau_{2,n}^{-1}\big\}\big]= R_1 \big[\big\{\tau_{0,n}^{-1}\big\}\big],
\end{gather*}
we have
\begin{gather*}
\Omega_{R_0[\{\tau_{2,n}^{-1}\}]}(\tau_{3,n})\le r
\end{gather*}
from Lemma \ref{UFDlemma}. Since $\tau_{3,n}$ is coprime with $\tau_{2,n'}$ for every integer $n'$, we have
\begin{gather*}
\Omega_{R_0} (\tau_{3,n})=\Omega_{R_0[\{\tau_{2,n}^{-1}\}]}(\tau_{3,n}).
\end{gather*}
Thus $\Omega_{R_0}(\tau_{3,n})\le r$.
On the other hand, from the expression
\begin{gather*}
\tau_{3,n}=\frac{1}{\tau_{1,n}}\prod_{j=1}^r \big( \tau_{2,n+1}^m \tau_{2,n-1}^l -\zeta^{2j-1} \tau_{2,n}^k \big),
\end{gather*}
we have $\Omega_{R_0}(\tau_{3,n})\ge r$. The equality~\eqref{tau3n} indicates that the decomposition of $\tau_{3,n}$ into prime elements is written in the form of
\begin{gather*}
\tau_{3,n}=u\times \prod_{j=1}^r \big( \tau_{2,n+1}^m \tau_{2,n-1}^l -\zeta^{2j-1} \tau_{2,n}^k \big),
\end{gather*}
where $u$ is a unit element in $R_0$. By eliminating the denominators from each factor of the r.h.s., we obtain $p_{3,n}^{(j)}$. The degree $\deg p_{3,n}^{(j)}$ is independent of $j$ from Lemma~\ref{degmlk}. Lastly we shall prove that $p_{3,n}^{(i)}$ and $p_{3,n'}^{(j)}$ are mutually coprime if $(i,n)\neq (j,n')$. It is sufficient to investigate the case of $n=n'$, since, if $n\neq n'$, there exists at least one variable that the two iterates $p_{3,n}^{(i)}$ and $p_{3,n'}^{(j)}$ do not share. From the definition of $p_{t,n}^{(i)}$, there exist Laurent monomials $h$ and $h'$ such that
\begin{gather*}
h p_{3,n}^{(i)} = \tau_{2,n+1}^m \tau_{2,n-1}^l -\zeta^{2i-1}\tau_{2,n}^k, \qquad
h' p_{3,n}^{(j)} = \tau_{2,n+1}^m \tau_{2,n-1}^l -\zeta^{2j-1} \tau_{2,n}^k.
\end{gather*}
Therefore
\begin{gather*}
h p_{3,n}^{(i)} - h' p_{3,n}^{(j)} =\big(\zeta^{2j-1}-\zeta^{2i-1}\big) \tau_{2,n}^k.
\end{gather*}
Here $\zeta^{2i-1}-\zeta^{2j-1}\neq 0$. Let us suppose that $p_{3,n}^{(i)}$ and $p_{3,n}^{(j)}$ are not mutually coprime. Since they are both irreducible, $\tau_{2,n}^k$ must be divisible by both of them. From the factorization of $\tau_{2,n}^k$, there exists an integer $i'$ such that $p_{3,n}^{(i)}=p_{2,n}^{(i')}$, which contradicts the induction hypothesis.

{\bf The case of $\boldsymbol{t\ge 4}$.} Let us remind ourselves that $\tau_{t,n}$ is coprime with every element in $\{\tau_{t-1,n}, \tau_{t-2,n}\}_{n\in\mathbb{Z}}$ from Lemma~\ref{coprimethree}.

Let us prove that
\begin{gather}
\Omega_{R_0}(\tau_{t,n}) = r. \label{tautn}
\end{gather}
First we prove that $\Omega_{R_0}(\tau_{t,n}) \ge r$. We have
\begin{gather*}
\tau_{t,n}=q_{t-2,n}\prod_{j=1}^r \frac{ \tau_{t-1,n+1}^m \tau_{t-1,n-1}^l-\zeta^{2j-1} \tau_{t-1,n}^k }{ p_{t-2,n}^{(r-j+1)} }.
\end{gather*}
From Lemma~\ref{dividelemma}, we have
\begin{gather*}
P_j:=\frac{ \tau_{t-1,n+1}^m \tau_{t-1,n-1}^l-\zeta^{2j-1} \tau_{t-1,n}^k }{ p_{t-2,n}^{(r-j+1)} }\in R_0.
\end{gather*}
From Lemma~\ref{degmlk}, $P_j$ is not a unit in $R_0$. Thus we have
\begin{gather*}
\Omega_{R_0}(\tau_{t,n})\ge r.
\end{gather*}
It is trivial that $\Omega_{R_{t-2}}(\tau_{t,n})=r$. Since we have
\begin{gather*}
R_0 \big[\big\{\tau_{t-2,n}^{-1}, \tau_{t-1,n}^{-1}\big\}_{n\in\mathbb{Z}}\big]=R_{t-2} \big[\big\{\tau_{0,n}^{-1}, \tau_{1,n}^{-1}\big\}_{n\in\mathbb{Z}}\big],
\end{gather*}
from the Laurent property of every iterate, we obtain
\begin{gather*}
\Omega_{R_0 [\{\tau_{t-2,n}^{-1}, \tau_{t-1,n}^{-1}\}_{n\in\mathbb{Z}}]}(\tau_{t,n}) \le r.
\end{gather*}
By using Lemma~\ref{UFDlemma}, we have
\begin{gather*}
\Omega_{R_0}(\tau_{t,n}) \le r.
\end{gather*}
Therefore we have proved the equality \eqref{tautn}.

From these observations, we conclude that the prime element decomposition of $\tau_{t,n}$ is of the following form:
\begin{gather*}
\tau_{t,n}=u\times \prod_{j=1}^r P_j ,
\end{gather*}
where $u$ is a unit element in $R_0$. The term $p_{t,n}^{(j)}$ is obtained by taking the numerator of $P_j$. The degree $\deg p_{t,n}^{(j)}$ is independent of $j$ from Lemma~\ref{degmlk}.

Lastly we shall prove that $p_{t,n}^{(j)}$ and $p_{s,n'}^{(i)}$ are coprime if $(t,n,j)\neq (s,n',i)$. First, if $2\le t\neq s$, two iterates are coprime since $\deg p_{t,n}^{(j)} \neq \deg p_{s,n'}^{(i)}$. Let us show that these two factors are coprime when $t=s$ and $(j,n)\neq (i,n')$. From the construction of $p_{t,n}^{(j)}$, there exist two Laurent polynomials $h,h'\in R_0$ such that
\begin{gather*}
h p_{t,n}^{(j)} = \tau_{t-1,n+1}^m \tau_{t-1,n-1}^l -\zeta^{2j-1} \tau_{t-1,n}^k, \qquad h' p_{t,n'}^{(i)} = \tau_{t-1,n'+1}^m \tau_{t-1,n'-1}^l -\zeta^{2i-1} \tau_{t-1,n'}^k.
\end{gather*}
If $n\neq n'$, there exists at least one variable that the two iterates $p_{t,n}^{(j)}$ and $p_{t,n'}^{(i)}$ do not share. Since the two iterates are both irreducible, they must be coprime. If $n=n'$, we have
\begin{gather*}
h p_{t,n}^{(j)} - h' p_{t,n}^{(i)} = \big( \zeta^{2i-1}-\zeta^{2j-1} \big) \tau_{t-1,n}^k.
\end{gather*}
Thus the coprimeness of the two iterates are obtained from the discussion same as in the case of $t=3$.
\end{proof}

\section{More examples}
From here on we introduce two more examples without the irreducibility but having the coprimeness property.
\subsection{Coprimeness-preserving Somos-4}
By a reduction from the $1$-dimensional CP Toda lattice equation \eqref{cp1dtoda-prev}, we obtain the following recurrence
\begin{gather}
x_n=\frac{x_{n-1}^{rm} x_{n-3}^{rl}+x_{n-2}^{rk}}{x_{n-4}}, \label{CPSomos4}
\end{gather}
which can be interpreted as the extension to the Somos-$4$ recurrence \cite{Gale}:
\begin{gather}
z_n z_{n-4}=z_{n-1} z_{n-3} + z_{n-2}^2. \label{somos4}
\end{gather}
Here, $r\ge 2$ and $\mathrm{GCD}(m,l,k)=1$. It is already obtained that the original Somos-$4$ has the Laurent property, the irreducibility and the coprimeness~\cite{FZ, dKdVSC2}. Each iterate $z_n$ of~\eqref{somos4} is proved to be in $\mathbb{Z}\big[z_0^{\pm},z_1^{\pm},z_2^{\pm},z_3^{\pm}\big]$ using a technique of the cluster algebras~\cite{FZ}. Moreover, the irreducibility and the coprimeness are proved in~\cite{dKdVSC2}. Applying the techniques introduced in this article to equation \eqref{CPSomos4}, we have the coprimeness property even when the irreducibility is no longer satisfied. Let us express $x_n=p_n/q_n$, where~$p_n$ is a polynomial, $q_n$ is a monic monomial and~$p_n$ and~$q_n$ are mutually coprime.
\begin{Theorem}\quad
\begin{enumerate}\itemsep=0pt
\item[$1.$] For $n\ge 4$, the numerator $p_n$ is factorized into $r$ non-unit irreducible polynomials in $\mathbb{C}[x_0,x_1,x_2,x_3]$ as
\begin{gather*}
p_n=p_n^{(1)}p_n^{(2)}\cdots p_n^{(r)}.
\end{gather*}
\item[$2.$]
\begin{gather*}
\deg p_n^{(1)}=\deg p_n^{(2)}=\cdots =\deg p_n^{(r)}=\frac{\deg p_n}{r}, \qquad n\ge 4.
\end{gather*}
\item[$3.$] By appropriately rearranging the order of $p_{n}^{(1)},\dots,p_{n}^{(r)}$, we have
\begin{gather*}
p_{n}^{(j)}\, \big|\, \big(x_{n-1}^{m} x_{n-3}^{l} - \zeta^{2j-1} x_{n-2}^{k}\big), \qquad n\ge 4,
\end{gather*}
in $R:=\mathbb{C}\big[x_0^{\pm}, x_1^{\pm}, x_2^{\pm}, x_3^{\pm}\big]$ for every $j$.
\item[$4.$] $p_n^{(j)}$ and $p_m^{(i)}$ coprime as polynomials if and only if $(n,j)\neq (m,i)$.
\end{enumerate}
\end{Theorem}
The proof is done in the same manner as in Theorem~\ref{mainthm}.

\subsection{Coprimeness-preserving two-dimensional Toda lattice}
Our last example in this paper is a generalization of the two-dimensional discrete Toda lattice equation
\begin{gather}
\tau_{t+1,n,m+1}\tau_{t-1,n+1,m}=\tau_{t,n+1,m}^{k_1}\tau_{t,n,m+1}^{k_2}+\tau_{t,n,m}^{l_1}\tau_{t,n+1,m+1}^{l_2},\qquad k_i, l_i \in \Z_+,
\label{pDToda_polinear_eq}
\end{gather}
which we studied in \cite{KKMT}.

\begin{Theorem} \label{factorizethm} Each iterate $\tau_{t,\n}$ of equation~\eqref{pDToda_polinear_eq} is a Laurent polynomial of the initial variables $\left\{\tau_{0, n,m},\,\tau_{1, n,m}\, |\, n,m\in\mathbb{Z}\right\}$. Moreover, every pair of the iterates is always co-prime.
\end{Theorem}

\looseness=-1 Proof is omitted here because the discussion is almost the same as our main Theorem~\ref{mainthm} for the one-dimensional case. We already showed this theorem under the condition that the right hand side of the equation is not factorizable. Note that the irreducibility of $P^{k_1}Q^{k_2}+R^{l_1}S^{l_2}$ in $\mathbb{Z}[P,Q,R,S]$ is equivalent to the condition that $\mathrm{GCD}(k_1,k_2,l_1,l_2)\neq 2^k$ with $k\ge 0$. Even if the irreducibility is not satisfied in Theorem~\ref{factorizethm}, the Laurent property and the coprimeness still hold.

\section{Conclusion}
In this paper we studied the coprimeness property of several discrete dynamical systems. We first explained our motivation using a simple recurrence relation with coprimeness property but without the irreducibility. Then, we proved the coprimeness property of an extension to the one-dimensional discrete Toda equation when the equation is factorizable. Finally we stated without proof the coprimeness property of extensions to the Somos-$4$ equation and the two-dimensional discrete Toda equation. In these examples, each iterate factorizes in exactly the same manner as the defining equation itself. In such cases it is possible to investigate the coprimeness property by following the factorization and formulating the evolution of each factor. The examples in this paper have the coprimeness property even if the equations themselves are factorizable and thus their iterates do not have the irreducibility property. Therefore, the coprimeness, which is based on a singularity analysis related to the singularity confinement test, can be individually investigated, while in our previous works the coprimeness was always paired with the irreducibility property.

\appendix

\section{Basic facts on unique factorization domains} \label{Append}
\begin{Definition}
Let $f\in\mathbb{C}\big[X_1^{\pm},X_2^{\pm},\dots \big]$ be a Laurent polynomial. Let us decompose $f$ as $f=gh$, where $g$ is a monic Laurent monomial, $h$ is a polynomial without any monomial factor. The degree of~$f$ as a Laurent polynomial is defined as
\begin{gather*}
\deg_L f:=\deg h.
\end{gather*}
\end{Definition}

\begin{Definition}
Let $R$ be a unique factorization domain. Let us factorize a non-zero element~$f$ in~$R$ into prime elements $p_i\in R$ as
\begin{gather*}
f=u p_1^{e_1}p_2^{e_2}\cdots p_m^{e_m},
\end{gather*}
where $u$ is a unit in $R$, and $e_i$ is a positive integer. We define the function $\Omega_R$ as
\begin{gather*}
\Omega_R(f)=e_1+\cdots +e_m,
\end{gather*}
which we will call the total number of prime elements of $f$ in $R$.
\end{Definition}
\begin{Lemma}\label{UFDlemma} Let $R$ be a unique factorization domain, and let us take two non-zero elements $f,g\in R$. Then we have
\begin{gather*}
\Omega_{R[g^{-1}]}(f)\le \Omega_R (f),
\end{gather*}
where the equality is satisfied if and only if $f$ and $g$ are mutually coprime in $R$.
\end{Lemma}

\begin{proof}
Let us factorize $f$ into prime elements $p_i\in R$ as
\begin{gather*}
f=u p_1^{e_1}p_2^{e_2}\cdots p_m^{e_m},
\end{gather*}
where $u$ is a unit in $R$, and $e_i$ is a positive integer. Let us rearrange the order of the terms $p_1,\dots, p_m$ so that there exists an integer $r$ such that $p_i\mathrel{\not |} g$, $1\le i\le r$, and $p_j\, |\, g$, $r+1\le j\le m$. Then the factorization of~$f$ in the localized ring $R\big[g^{-1}\big]$ is
\begin{gather*}
f=vp_1^{e_1}p_2^{e_2}\cdots p_r^{e_r},
\end{gather*}
where $v=u p_{r+1}^{e_{r+1}}\cdots p_m^{e_m}$ is a unit in $R[g^{-1}]$.
\end{proof}

\subsection*{Acknowledgements}
We thank the referees for reminding us several important papers regarding the Laurent systems and the discrete integrability. We acknowledge support from KAKENHI Grant numbers 26400109, 16H06711 and 17K14211.

\pdfbookmark[1]{References}{ref}
\LastPageEnding

\end{document}